\documentclass[conference]{IEEEtran}
\IEEEoverridecommandlockouts

\usepackage{cite}
\usepackage{amsmath,amssymb,amsfonts}
\usepackage{algorithmic}
\usepackage{graphicx}
\usepackage{textcomp}

\usepackage{lipsum,framed}
\usepackage{tcolorbox}
\usepackage{amsmath}
\usepackage{bm}
\usepackage{float}
\usepackage[normalem]{ulem}

\usepackage{url}
\usepackage[colorlinks=true, linkcolor=blue, citecolor=blue, urlcolor=blue]{hyperref}

\usepackage[margin=1in]{geometry}
\usepackage{array}
\usepackage{booktabs}
\usepackage{multirow}

\usepackage{amsthm}
\usepackage{amssymb}

\usepackage{fancyhdr}

\usepackage{colortbl}      
\usepackage{caption}
\newtheorem{theorem}{Theorem} 
\newtheorem{corollary}{Corollary}
\newtheorem{lemma}{Lemma}

\newtheoremstyle{case}{}{}{}{}{}{:}{ }{}
\theoremstyle{case}

\usepackage{pgfplots}
\pgfplotsset{compat=1.16}

\usepackage{footmisc}
\usepackage{caption}
\usepackage{lineno}

\usepackage{enumitem}
\newlist{properties}{enumerate}{2}
\setlist[properties]{label=Property \arabic*., font=\textbf, itemindent=*}

\expandafter\def\expandafter\normalsize\expandafter{%
    \normalsize%
    \setlength\abovedisplayskip{0pt}%
    \setlength\belowdisplayskip{8pt}%
    \setlength\abovedisplayshortskip{2pt}%
    \setlength\belowdisplayshortskip{2pt}%
}

\def\BibTeX{{\rm B\kern-.05em{\sc i\kern-.025em b}\kern-.08em
    T\kern-.1667em\lower.7ex\hbox{E}\kern-.125emX}}

\fancypagestyle{FirstPage}{
\lfoot{\footnotesize \copyright2026 IEEE. Personal use of this material is permitted. Permission from IEEE must be obtained for all other uses, in any current or future media, including reprinting/republishing this material for advertising or promotional purposes, creating new collective works, for resale or redistribution to servers or lists, or reuse of any copyrighted component of this work in other works.} 
}
    
\begin{document}


\newcommand{\OC}[2]{\ensuremath{OC_{#1,#2}}}

\newcommand{\DevNG}[2]{\ensuremath{Deviation_{NG,#1,#2}}}

\newcommand{\CF}[4]{\ensuremath{CF_{#2}(#1,x_#3,...,x_#4)}}

\newcommand{\DevNN}[3]{\ensuremath{Deviation_{NN,#1#2,#3}}}

\newcommand{\PiNN}[0]{\ensuremath{\pi_{NN}}}

\newcommand{\PiNG}[0]{\ensuremath{\pi_{NG}}}

\newcommand{\OCwc}[1]{\ensuremath{OC_{#1,WC}}}

\newcommand{\maxDrift}[0]{\ensuremath{\rho}}

\newcommand{\DriftNode}[0]{\ensuremath{\rho_N}}

\newcommand{\DriftGuardian}[0]{\ensuremath{\rho_G}}

\newcommand{\AccuracyCF}[0]{\ensuremath{\alpha(\delta)}}

\newcommand{\tinj}[0]{\ensuremath{T_{INJ}}}

\newcommand{\DevDrift}[1]{\ensuremath{\Delta_{Drift,#1}}}

\newcommand{\DevTinj}[1]{\ensuremath{\Delta_{INJ,#1}}}

\newcommand{\Resync}[0]{\ensuremath{R_{int}}}

\newcommand{\CFGen}[2]{\ensuremath{CF_{#2}(#1)}}

\newcommand{\DriftDirRatio}[0]{\ensuremath{R_{drift}}}

\newcommand{\PiCF}[0]{\ensuremath{\pi_{conv}}}

\newcommand{\DeltaCF}[0]{\ensuremath{\delta}}

\newcommand{\DriftRelative}[2]{\ensuremath{\rho_{#1#2}}}

\newcommand{\DevTotal}[1]{\ensuremath{\Delta_{#1}}}

\newcommand{\TStartSlot}[2]{\ensuremath{T_{#1,#2}}}

\newcommand{\TStartCycle}[1]{\ensuremath{T_{START,#1}}}

\newcommand{\ReqSM}[0]{\ensuremath{\Delta_{SM}}}

\newcommand{\ReqSMMax}[0]{\ensuremath{\Delta_{SMMax}}}

\newcommand{\SlotMTs}[0]{\ensuremath{l_{slot}}}

\newcommand{\CorrSegMTs}[0]{\ensuremath{l_{corr}}}

\newcommand{\NSlotSM}[0]{\ensuremath{N_{slot}}}

\newcommand{\NCorrSegSM}[0]{\ensuremath{N_{corr}}}

\newcommand{\NTarget}[0]{\ensuremath{N_{TC}}}

\title{Exploring and Exploiting Synchrony Limitations of Time-Triggered Network-Agnostic Guardians
}

\author{
  \IEEEauthorblockN{
    Shreya Vithal Kulhalli, Mohammad Ibrahim Alkoudsi, Gerhard Fohler
    }
  \IEEEauthorblockA{
    RPTU Kaiserslautern-Landau\\
    Kaiserslautern, Germany \\
    \{kulhalli,alkoudsi,gerhard.fohler@rptu.de\}
    }
}

\maketitle

\begin{abstract}
Time-triggered communication protocols rely on trusted components known as guardians to enforce adherence to predetermined network schedules. 
Network-agnostic guardians offer an efficient and scalable distributed solution
with reduced implementation cost and complexity compared to network‑aware alternatives. 
However, this efficiency is based on the guardian’s dependence on the controlled node for clock synchronization, which introduces a vulnerability: a malicious node can exploit this dependency to launch timing attacks against its guardian and eventually interfere with messages from other nodes on the network.

In this paper, we establish a theoretical lower bound on the attainable clock synchronization precision between a node and its network-agnostic guardian. Building on this result, we introduce a timing attack that leverages the unavoidably imperfect clock synchrony to cause controlled and undetected de-synchronization of the guardian. The attack enables a malicious node to cause collisions with targeted critical network messages. 
We evaluate the effectiveness of the attack using a FlexRay field bus network model implemented in the OMNeT++ simulation framework. Our results show that the attack is able to remain undetected with 100\% success and disrupts the transmission of the critical messages of the target node by causing collisions with them with 100\% success.
\end{abstract}


\thispagestyle{FirstPage}

\section{Introduction and Related Work}\label{sec:intro}
Time-triggered communication protocols (e.g., TTP~\cite{TTPStandard2011}, FlexRay~\cite{FlexRay2008}, TTEthernet~\cite{TTEthernetStandard2016}) are designed to efficiently deliver the determinism and fault tolerance required by safety-critical systems~\cite{Kopetz2003_TTA}. These protocols provide timely network access by establishing a global time base through bounded-precision clock synchronization among nodes and by organizing network accesses according to a \emph{Time-Division Multiple Access (TDMA)} strategy. In this scheme, each node is assigned exclusive slots on the channel and uses the full medium bandwidth during its slots. System designers compute the TDMA schedule offline and introduce gaps between consecutive slots to account for the unavoidably imperfect clock synchronization among nodes, preventing inter-node transmission overlap (i.e., collisions).

Fault isolation, and thus the preservation of determinism, is strengthened by bus guardians -- trusted hardware components connected to each node that enforce adherence to the specified communication schedule~\cite{Temple98_TemporalIsolation_Guardian,Bauer03_TWA,FlexRayBG2008,TTEthernetServices2009}. To contain faulty behavior, such as nodes attempting to transmit outside their TDMA slot boundaries, guardians typically require knowledge of the global time and the TDMA schedule. Achieving fault independence further demands that guardians resynchronize their local clocks to global time on their own. Consequently, contemporary \emph{network-aware} guardians incorporate protocol- and topology-specific mechanisms, which increases both cost and complexity. 

In contrast, network-agnostic guardians avoid protocol-specific logic. Originally proposed under constrained fault assumptions~\cite{Temple98_TemporalIsolation_Guardian}, these guardians rely on their controlled node for clock synchronization to infer schedule alignment, significantly cutting complexity and cost. However, modern systems increasingly operate in harsher environments over extended mission durations, facing higher fault rates and targeted cyber-physical attacks~\cite{ca-powerGridAttack2016,ca-waterBreach2008,ca-dronHack2012,ca-killjeep2015,ca-steelMill2014,ca-steelMill2015,ca-stuxnetLessons2011,ca-aurora2018,ca-hacksatellite2020}.

Our paper assumes the presence of adversaries capable of inferring timing information about the target system. This assumption is well supported by significant prior research that demonstrates the plausibility of this threat. Pellizzoni et al.~\cite{DBLP:conf/rtas/PellizzoniPYBMB15}, Chen et al.~\cite{DBLP:conf/rtas/ChenMPBK19}, and Liu et al.~\cite{DBLP:journals/access/LiuY20b} investigate threats emerging from covert side channels that adversaries can exploit to extract crucial timing information through timing-inference attacks. This assumption has also been adopted in methods targeting task scheduling rather than communication. Yoon et al.~\cite{DBLP:conf/rtas/YoonMCS16} introduce a schedule-obfuscation method that randomizes schedules while maintaining timing guarantees. Chen et al.~\cite{DBLP:journals/corr/abs-1806-01393} propose a protocol for obscuring predictable timing behavior in real-time systems using dynamic-priority scheduling. Krüger et al.~\cite{DBLP:conf/ecrts/KrugerVF18} pursue a similar motivation and investigate online and offline approaches for schedule randomization in time-triggered systems to mitigate timing-inference attacks.

Recent research has revealed that attacks compromising the \emph{synchronization protocol} pose a critical threat to determinism and fault isolation. For instance, PCSPOOF~\cite{Loveless2023} demonstrated the first successful breach of isolation in \emph{Time‑Triggered Ethernet (TTE)}, causing synchronization loss and dropping critical messages. This threat landscape reveals a critical vulnerability in network-agnostic guardians: their reliance on node-provided timing makes them susceptible to manipulation by a compromised node. In~\cite{syncguard2023}, Alkoudsi et al. demonstrated that under malicious threats, the dependency on the controlled node for synchronization enables time-domain attacks such as omission of relayed synchronization messages, replaying messages, and causing asymmetric reception that hinders fault attribution. They further proved that these attacks persist even with authenticated messaging.  

In this paper, we design an advanced replay attack that enables the malicious node to target and disrupt the transmission of a specific message in the schedule. The attack is based on the same principle as the replay attack in \cite{syncguard2023}, where the compromised node gradually introduces arbitrarily small delays into the synchronization signal sent to its guardian, subtly altering its local perception of global time.
In order to increase the effectiveness of the attack (as fast and stealthy as possible), we first derive a theoretical lower bound on the attainable precision of external clock synchronization between a node and its network-agnostic guardian, then compute the maximum covert per-resynchronization round injection delay that remains within this bound.  

By tracking global time progression and repeatedly applying such delay injections, the malicious node shifts the guardian’s view on global time until its slot aligns with the victim's slot. During the victim's slot, the malicious node transmits, causing a collision that invalidates the target critical message.

We evaluate our attack in the OMNeT++ network simulation library \cite{omnetppsimmanual,Virdis19_OmnetppBook} with the FiCo4OMNeT FlexRay model \cite{FiCo4OMNeT} on an eight‑node network, with each node extended with a network‑agnostic guardian. The simulation results confirm the attack’s effectiveness in disrupting messages in the target slot with $100\%$ success rate, underscoring the importance of coordination among network‑agnostic guardians and enforcing node rejuvenation when malicious behavior cannot be precluded, as discussed in \cite{syncguard2023}.

This work presents the following contributions:
\begin{enumerate}
    \item Precision bound: we derive a lower bound on the attainable precision of external clock synchronization between a node and its network-agnostic guardian.
    \item Target slot attack: we design a time-domain attack against network-agnostic guardians that causes collisions with a target critical message sent in a specific slot in the TDMA schedule.
    \item Experimental evaluation: we validate the attack using fault injection in a simulated FlexRay network (using OMNeT++ and FiCo4OMNeT's FlexRay model) augmented with network-agnostic guardians.
\end{enumerate}
  
The remainder of the paper is organized as follows. We state our system assumptions in Section~\ref{sec:system_model} and subsequently cover some essential background on synchronizing clocks to establish a global notion of time in Section~\ref{sec:background}. In Section~\ref{sec:ngprecision}, we derive a lower precision bound for external node–guardian synchronization, and exploit this bound, in Section~\ref{sec:attack}, to design our target slot attack. We present our experimental evaluation in Section~\ref{sec:evaluation}. We propose countermeasures that can be adopted to mitigate the target slot attack in Section~\ref{sec:countermeasures} before drawing our conclusions in Section~\ref{sec:conclusion}.

\section{System Model}\label{sec:system_model}

We consider real-time distributed systems, as shown in Figure~\ref{fig:system_model}, which comprise $N$ nodes that communicate by sending messages over a shared time-triggered network channels. Although we do not constrain our system to a specific time-triggered network protocol and topology, we assume that access to the network follows a TDMA strategy based on a predetermined schedule and a global notion of time. 
The schedule is organized into communication cycles, each consisting of a sequence of time slots, where a time slot is a fixed interval of global time that is uniquely identified by a (\emph{Slot ID}) and is assigned exclusively to at most one node. 
We further assume that each communication cycle ends with a dedicated slot or segment for clock correction. During this segment, no communication occurs. Instead, all correct nodes adjust their perception of global time by proportionally stretching the duration of the clock correction segment, based on the estimated differences between their local clocks. 

We assume that clock drift remains bounded by a known value \maxDrift{}, and denote local clock granularity by \emph{microtick ($\mu$T)}. We further assume that all nodes synchronize their local clocks, whether through internal synchronization mechanisms or by relying on an external time source, within a bounded precision \PiNN{}, for example, by using the fault-tolerant averaging algorithm \cite{Kopetz87_ClockSync}. This synchronization underpins the global time base between nodes, which ticks with coarser granularity than $\mu$T, termed \emph{macrotick (MT)}.

As shown in Figure \ref{fig:system_model},
each node incorporates a trusted, network‑agnostic guardian component that determines when the node is permitted to transmit. It enforces this by enabling or disabling the node’s send interface for a given communication channel through the \emph{transmission controller (TC)}. The guardian maintains its own notion of global time through an internal local clock, which it synchronizes with global time based on a periodic external clock synchronization (CS) signal emitted by the controlled node through a dedicated physical wire. The CS signal indicates the start of a communication cycle from the node's perspective. 

We assume that guardians perform a rudimentary acceptance window check on the CS signal, considering it valid only if it arrives within a time window centered at the expected arrival time of this signal from the guardian's perspective. The acceptance window has a minimum duration of 2\PiNG{}, accounting for potential clock drift in both directions. Here, \PiNG{} represents the synchronization precision between the node and its guardian. When accepted, the guardian aligns its clock with that of the node.

\begin{figure}[h]
  \centering
  \includegraphics[width=0.48\textwidth]{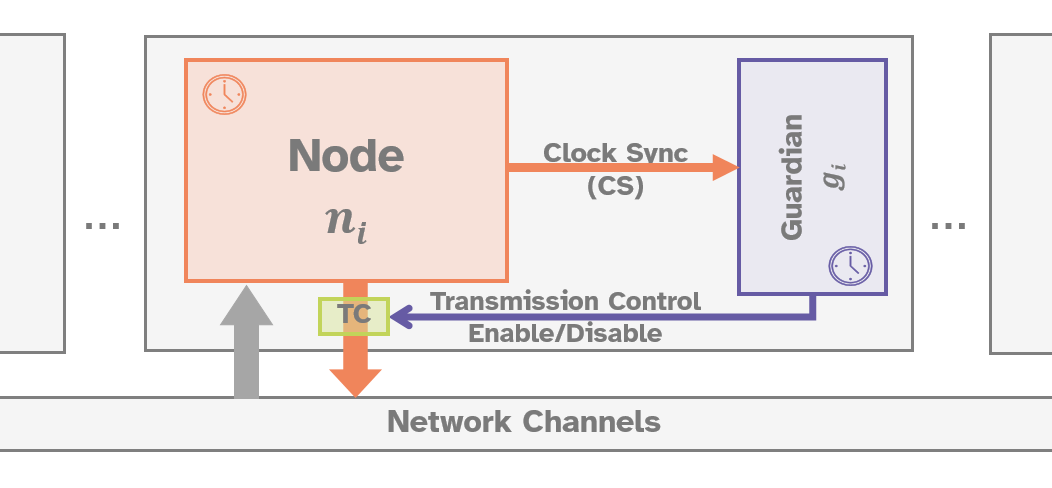}
  \caption{System Model}
  \label{fig:system_model}
\vspace{-0.5cm}\end{figure}

We adopt the same failure model as in~\cite{syncguard2023}, in which nodes may exhibit arbitrary failures, whereas guardians and TCs are limited to crash failures. When a guardian crashes, it permanently disables its node’s ability to transmit on the network, thereby enforcing fail‑silence. This behavior follows from the guardian’s role as a trusted component and from the inherent simplicity of network‑agnostic guardian designs.

\section{Background}\label{sec:background}
\subsection{Clock Synchronization and Global Time}\label{subsec:clocksync_gtb}

Time-triggered communication follows a TDMA schedule, defined using the granularity of the shared global time base. In order to correctly follow the schedule, all nodes must have a consistent notion about global time progression. This includes agreement on the duration of the $MT$, denoted as $g_{MT}$~\cite{Obermaisser2011_TTC,DBLP:KopetzS22_Clocks_Time}.

Internally synchronizing the local clocks of nodes to maintain consistent view of global time progression typically involves exchanging clock values, computing pairwise clock differences, and applying a \emph{convergence function (CF)} to bring clocks together. 
The CF is characterized by the following properties~\cite{Schneider1987_ByzClockSync}: 
\begin{enumerate}
    \item \textbf{Precision Enhancement $\pi_{conv}$}: The parameter, \emph{precision of CF}, quantifies how closely the clocks of correct nodes are synchronized after applying the CF. The CF achieves this by imposing an upper bound on the clock difference $\delta$ between any pair of correct nodes, such that $\delta \leq \PiNN{}$, where $\PiNN{}$ denotes the precision of internal clock synchronization. In order for the CF to successfully converge the clocks, the inequality, $\pi_{conv} < \delta$ has to be satisfied. 
    \item \textbf{Accuracy Preservation $\alpha(\delta)$}: The parameter, \emph{accuracy of CF}, defines the maximum permissible adjustment that a node may apply to its clock during a resynchronization cycle. The adjustment is constrained by the inequality $\alpha(\delta) \leq \delta$, ensuring that clock corrections remain within the bounds of the observed clock differences. 
\end{enumerate}

In TT communication protocols, each node typically adjusts its perception of global time by recalibrating the macrotick duration $g_{MT}$, based on the value produced by the convergence function. 
$g_{MT}$ is defined as an integer multiple of the node’s local clock granularity $\mu T$. 
The adjusted macrotick duration remains in effect until the subsequent resynchronization cycle.

\section{Precision of Node-Guardian Clock Synchronization}\label{sec:ngprecision}
Now, we turn to deriving a safe lower bound for the synchronization precision between a node and its network-agnostic guardian, under the assumptions stated in Section~\ref{sec:system_model}.

There are two primary sources that contribute to the clock deviation between a node and its network-agnostic guardian: 
\begin{enumerate}
    \item \textbf{Clock drift}: Refers to the gradual divergence between the node's and guardian's clocks due to the differences in oscillator frequencies. 
    Let $c_d$ represent the clock deviation between the guardian and its node due to clock drift within a single resynchronization cycle. 
    Since drift may occur in both directions, the maximum value of $c_d$ is $2\maxDrift{}\Resync{}$, 
    where \maxDrift{} denotes the upper bound on the clock drift rate, and \Resync{} is the duration of the resynchronization cycle. 
    This scenario occurs when the node experiences the maximum negative drift rate and the guardian the maximum positive drift rate. 
    Conversely, the minimum value of $c_d$ is $-2\maxDrift{}\Resync{}$, 
    representing the scenario in which the node has the maximum positive drift rate and the guardian the maximum negative drift rate. 
    \item \textbf{Node clock correction}:
    The internal clock synchronization of the node's clock to global time affects the difference between its clock and that of the guardian's. 
    Let $c_n$ denote the CF value applied by the node during a resynchronization cycle. 
    The value of $c_n$ lies within the range $[-\alpha(\delta), +\alpha(\delta)]$, where $\alpha(\delta)\leq\delta\leq\PiNN{}$ (See Sec.\ref{subsec:clocksync_gtb}). $c_n$ reaches the boundary values in scenarios where the node experiences the maximum possible clock drift in the opposite direction relative to all other correct node clocks, and consequently, relative to global time.
\end{enumerate}

\begin{lemma}
\label{lem:extern_clock_sync}
For any resynchronization cycle, the correction value $c_g$ applied to the guardian's clock upon accepting a CS signal by its node is determined by:
\begin{equation}
    c_g = c_n - c_d
    \label{eq:ng_cs}
\end{equation}
where $c_n$ is the correction value applied to the node's clock during internal synchronization and $c_d$ is the deviation due to clock drift accumulated between the node and the guardian within a resynchronization cycle. 
\end{lemma}

\begin{proof}

Let $t_{cs}^N$ denote the start-of-cycle time according to the node's clock before applying the CF. After applying the correction value $c_n$, the updated start-of-cycle time becomes $t_{cs'}^N = t_{cs}^N + c_n$.

Let $t_{cs}^G$ denote the start-of-cycle time according to the guardian's clock, then we can write: 
\[ t_{cs}^G = t_{cs}^N + c_d\]

The correction value $c_g$ is computed as the difference between the corrected time received from the node and the expected start-of-cycle time at the guardian:

\[
c_g = t_{cs'}^N - t_{cs}^G = (t_{cs}^N + c_n) - (t_{cs}^N + c_d) = c_n - c_d
\]
    
\end{proof}

From Lemma~\ref{lem:extern_clock_sync}, a safe clock precision 
for the externally synchronized network-agnostic guardians must be equal to or greater than the maximum possible value of $c_g$. 
Theorem~\ref{thm:c_g_upper_bound} demonstrates that
$c_g$ is bounded from above by the precision achieved through the internal clock synchronization between nodes \PiNN{}.
\begin{theorem}
\label{thm:c_g_upper_bound}
For any resynchronization cycle, the correction value $c_g$ applied to the guardian's clock upon accepting a CS signal by its node is bounded by: 

\begin{equation}
    |c_g| \leq \PiNN{}
    \label{eq:theorem_2}
\end{equation}

\end{theorem}

\begin{proof}
Considering Equation~\ref{eq:ng_cs}, the maximum value of $c_g$ occurs when $c_n$ reaches its maximum positive value and $c_d$ reaches its maximum negative value. 

The maximum $c_n$ occurs when the node's clock has the maximum negative drift rate relative to all other nodes (representing global time), which simultaneously exhibit the maximum positive drift rate. In addition, messages between nodes are subjected to the worst-case jitter and faults. Since $\alpha(\delta)\leq\delta\leq\PiNN{}$, in the worst-case we have  $c_n=\PiNN{}$. 
In contrast, the minimum $c_d$ occurs when the node's clock exhibit the maximum positive drift rate relative to the guardian's clock, while the guardian's clock simultaneously experiences the maximum negative drift rate, yielding: $c_d=2\maxDrift{}\Resync{}$
Since the node's clock cannot possibly have the maximum positive and negative drift rates simultaneously, i.e., these scenarios are mutually exclusive, we have $|c_g| < 2\maxDrift{}\Resync{} + \PiNN{}$. 

In Table~\ref{tab:devCases}, we present the extreme cases of clock drift, illustrating opposing deviations among the node, the guardian, and the global time (collectively representing all other nodes in the system). 
For each case, the corresponding values of $c_d$, $c_n$, and $c_g$ are derived. 
In \emph{Case C}, where the node and global time exhibit drift in the same direction, the value of $c_n$ is determined solely by additional sources of uncertainty in the network such as communication jitter and faults, whose effects are bounded by the interval $|\PiNN{}-2\maxDrift\Resync|$ based on the design assumptions of internal clock synchronization. 

Since all the extreme drift cases, shown in Table~\ref{tab:devCases}, satisfy $-\PiNN{} \leq c_g \leq \PiNN{}$, Equation~\ref{eq:theorem_2} holds.

\begin{table}[ht!]
\centering
\setlength{\arrayrulewidth}{0.2mm}
\setlength{\tabcolsep}{2pt}
\renewcommand{\arraystretch}{0.9}
\resizebox{0.5\textwidth}{!}{
\begin{tabular}{|c|l|l|l|l|l|}
\hline
\textbf{Cases} & \textbf{Clock} & \textbf{Clock Drift} & $c_d$ & $c_n$ & $c_g$ \\
\hline
\multirow{3}{*}{A} 
  & Node        & $\mp\maxDrift{}$ & \multirow{3}{*}{$0$} & \multirow{3}{*}{$\pm\PiNN{}$} & \multirow{3}{*}{$\pm\PiNN{}$} \\
  & Guardian    & $\mp\maxDrift{}$ &     &         &         \\
  & Global time & $\pm\maxDrift{}$ &     &         &         \\
\hline
\multirow{3}{*}{B} 
  & Node        & $\mp\maxDrift{}$ & \multirow{3}{*}{$\pm 2\maxDrift \cdot \Resync$} & \multirow{3}{*}{$\pm\PiNN$} & \multirow{3}{*}{$\pm\PiNN \mp 2\maxDrift \cdot \Resync$} \\
  & Guardian    & $\pm\maxDrift{}$ &     &         &         \\
  & Global time & $\pm\maxDrift{}$ &     &         &         \\
\hline
\multirow{3}{*}{C} 
  & Node        & $\pm\maxDrift{}$ & \multirow{3}{*}{$\mp2\maxDrift \cdot \Resync$} & \multirow{3}{*}{$\pm\PiNN \mp 2\maxDrift \cdot \Resync$} & \multirow{3}{*}{$\pm\PiNN$} \\
  & Guardian    & $\mp\maxDrift{}$ &     &         &         \\
  & Global time & $\pm\maxDrift{}$ &     &         &         \\
\hline
\end{tabular}
}
\caption{Extreme cases of clock drift among the node, the guardian, and the global time, and corresponding values for $c_d$, $c_n$ and $c_g$.}
\label{tab:devCases}
\vspace{5pt} 
\end{table}
\vspace{-2em}

\end{proof}

\begin{corollary}
    The synchronization precision between a node and its network-agnostic guardian satisfies: 
    \begin{equation}\label{eq:precision_ng_lbound}
    \PiNG{} \geq \PiNN{}
    \end{equation}
\end{corollary}
The proof follows directly from Theorem~\ref{thm:c_g_upper_bound}, as the maximum guardian clock correction value establishes this precision bound. Consequently, the inter-slot gaps in the TDMA schedule must be determined with respect to \PiNG{} rather than \PiNN{}. 

Note that $\PiNG{} = \PiNN{}$ represents the best-case scenario for network-agnostic guardians. Depending on the physical deployment, however, additional factors such as CS signal propagation delay and processing jitter can result in $\PiNG{} > \PiNN{}$.
This has direct implications for feasibility of the target slot attack: Any timing attack must operate within the $\PiNG{}$ bound to avoid detection.

\section{Target Slot Attack}\label{sec:attack}
\subsection{Attacker Model}\label{subsec:attack_assumptions}

We consider an intelligent and malicious adversary capable of assuming complete control over a node, including sending messages over the network and performing clock synchronization. 
The adversary correspondingly has access to the same information as the node itself, 
including critical timing parameters associated with clock synchronization (precision, maximum drift rate, resynchronization cycle) and the communication schedule. 
Hence, the adversary knows which slots are designated for the transmission of critical messages. These assumptions are supported by analyses of adversaries capable of inferring crucial timing information through timing inference attacks \cite{DBLP:conf/rtas/ChenMPBK19,DBLP:journals/access/LiuY20b,DBLP:conf/ecrts/KrugerVF18}, as discussed in Section~\ref{sec:intro}.

The primary objective of the adversary is to disrupt the transmission of a critical message in the schedule. This objective can only be achieved by manipulating the global time perception of the guardian of the compromised node such that its assigned TDMA slot coincides with that of the critical message.
Achieving this objective requires the adversary to only evade detection by the guardian.

\subsection{Attack Procedure}\label{subsec:attack_proc}

The procedural flow of the target slot attack is presented in Figure~\ref{fig:aw_attack_flow}. 
In accordance with \cite{syncguard2023}, the compromised node initiates the attack by creating a \emph{virtual clock} that is initialized with the value of the compromised node's clock and remains unsynchronized with the global time. 
The virtual clock is employed by the node to monitor the progression of the attack. 
Afterwards, the node halts message transmission on the network, 
thereby leading other nodes in the cluster to infer its fail-silence. 
Note that this step is optional, as the objective is only to cause a collision with a specific message. 
Subsequently, the node enters a repetitive cycle in which, 
it computes and applies the maximum delay that can be covertly introduced with the CS signal $T_{inj}$ during each resynchronization cycle (described in detail in Sec.~\ref{subsec:calc_tinj}), 
followed by checking node's slot overlap with the target slot with reasonable accuracy (considering precision \PiNN{}). 
Once slot overlap is achieved, the node transmits a message on the network that collides with and invalidates the targeted critical message. 

\begin{figure}[h]
  \centering
  \includegraphics[width=.5\textwidth]{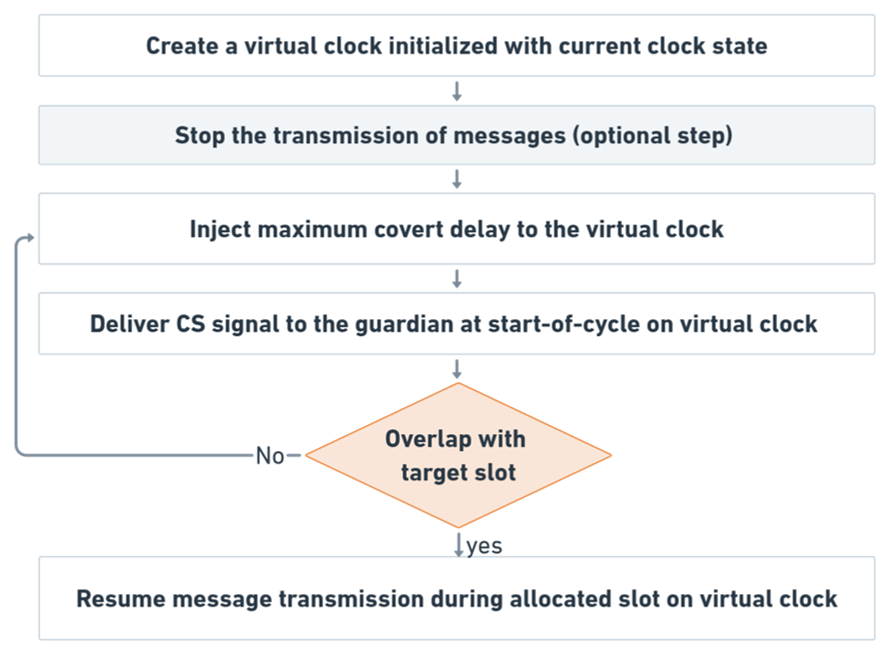}
  \caption{The procedural flow of the target slot attack.}
  \label{fig:aw_attack_flow}
\vspace{-0.6cm}\end{figure}

\subsection{Maximum Covertly Injected Delay}\label{subsec:calc_tinj}

The maximum delay that the node may inject when delivering the CS signal to its guardian without being detected by its acceptance window check is limited by the maximum clock deviation between their clocks within a resynchronization cycle, that is $2\maxDrift{}\Resync{}$, and the clock synchronization precision between the node and its guardian $\PiNG{}$.

Consider the example in Figure~\ref{fig:tinj_val}, which illustrates one communication cycle (\emph{cycle i}) of duration $\Resync{}$. In this cycle, \emph{node A} and its network-agnostic \emph{guardian A} experience the maximum clock drift rates $-\maxDrift{}$ and $+\maxDrift{}$, respectively. As a result, the node and the guardian accumulate a clock deviation of $2\maxDrift{}\Resync{}$. At the end of \emph{cycle i}, \emph{node A} injects a delay of \tinj{} together with the CS signal to \emph{guardian A}. 
In order to pass the guardian's acceptance window check, $\tinj{}$ must satisfy the inequality: $2\maxDrift{}\Resync{} + \tinj{} \leq \PiNG{}$. 
Reordering the terms, we obtain:

\begin{gather}
    \tinj{} = \PiNG{} - 2\maxDrift{}\Resync{}
    \label{eq:tnj_max}
\end{gather}

Since $\PiNG{} > 2\maxDrift{}\Resync{}$ the node is always able to covertly inject delays along with the CS signal. 
Note that injecting larger values than in Eq.~\ref{eq:tnj_max} may also go undetected depending on the clock deviation between the node and its guardian. However, since the actual clock drift values are uncertain, Eq.~\ref{eq:tnj_max} represents an upper bound on the delay that can be covertly introduced during each resynchronization cycle regardless of the clock deviation between the node and its guardian. 

\begin{figure}[h]
  \centering
  \includegraphics[width=.49\textwidth]{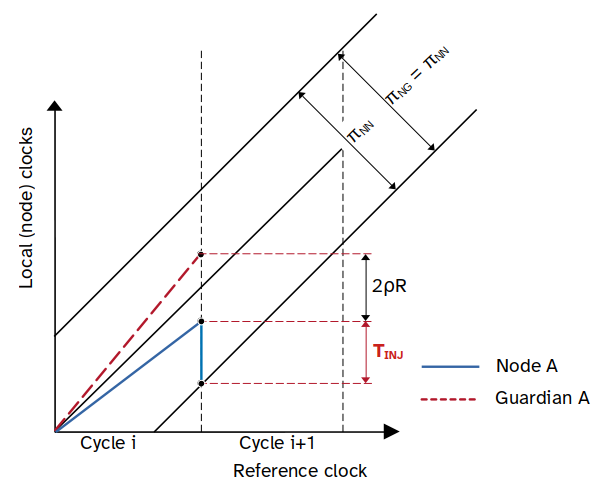}
  \caption{Illustration of the maximum delay \tinj{} that can be covertly introduced during each resynchronization cycle, independent of the actual clock deviation between the node and the guardian.}
  \label{fig:tinj_val}
\vspace{-1cm}\end{figure}

\subsection{Required Clock Deviation for Slot Overlap}\label{subsec:reqSM_nTC}

The clock deviation required (at the time of initiation of the attack) for an attacker’s slot to overlap with that of a victim node is influenced by the relative positioning of these two slots within the communication cycle and the direction of the injected delay.

For a given direction of injected delay, 
the necessary clock deviation is determined by evaluating the portion of the communication cycle that must be traversed.
For instance, given a TDMA schedule comprising 10 slots, if the attacker occupies Slot 5 and the victim occupies Slot 4, a positive injected delay necessitates traversing nine slots plus the correction segment, whereas a negative delay requires traversing only one slot. Consequently, the attacker selects the delay direction that minimizes the required clock deviation, thereby accelerating the attack's execution.

The required clock deviation for slot overlap in macroticks (MT) is given by:
    \[\NSlotSM{}\SlotMTs{} + \NCorrSegSM{}\CorrSegMTs{} \pm 1\]
    
where,
\NSlotSM{} denotes the \emph{number of slots to be traversed},
\SlotMTs{} represents the slot size in MT. 
\NCorrSegSM{} denotes the \emph{correction segment count} (set to 1 if a correction segment must be traversed and 0 otherwise).
$\CorrSegMTs{}$ represents the size of the correction segment in MT.
The additional one accounts for inherent imperfections in clock synchronization, with a positive sign for positive injected delays and a negative sign for negative delays.

It is important to note that, before runtime, the attacker cannot precisely predict when slot overlap with the victim node will occur, since it cannot forecast future clock differences. Consequently, in each resynchronization cycle, the attacker must rely on the current clock deviation relative to the victim node as an indicator, and iteratively compare it to the required clock deviation until slot overlap is successfully achieved. Since we already account for imperfect synchronization in the required clock deviation, identification of slot overlap is accurate within $\pm \PiNN{}$. Therefore, once slot overlap is identified, it is certain that the message sent by the attacker during its slot will collide with the victim node's message.

\section{Experimental Evaluation}\label{sec:evaluation}
In this section, we evaluate the effectiveness of the target slot attack through simulation-based experiments. The experimental configuration employs the OMNeT++ simulation library~\cite{omnetppuserguide}, augmented with the FiCo4OMNeT~\cite{FiCo4OMNeT,Buschmann13_FiCo4OMNeT} extension, to instantiate a FlexRay network comprising eight nodes (\emph{FRNode})s that communicate in accordance with the TDMA schedule presented in Table~\ref{tab:nwschedule} via the FlexRay bus (\emph{FRBus}). One execution of the TDMA schedule represents a communication cycle that ends with the clock correction segment, also known in the FlexRay protocol as the \emph{network idle time (NIT)} segment.

We extend each \emph{FRNode} with a custom \emph{guardian} submodule that follows the structure and assumptions specified in Section~\ref{sec:system_model}. 
Additionally, we implement a module, \emph{FRAttackerNode}, to represent a compromised node that performs the target slot attack. 
In the instantiated network, \emph{unitE} represents the compromised node, which is allocated the 5th slot in the communication cycle.  

\begin{table}[h!]
\centering
\arrayrulecolor{gray} 
\setlength{\arrayrulewidth}{0.2mm}
\setlength{\tabcolsep}{2.5pt}
\renewcommand{\arraystretch}{0.9}
\begin{tabular}{|c|c|c|c|c|c|c|c|c|c|c|c|}
\hline
\textbf{Slot ID} & 1 & 2 & 3 & 4 & \textbf{5} & 6 & 7 & 8 &  \multirow{2}{*}{\textbf{NIT}} \\
\cline{1-9}
\textbf{Node} & unit1 & unit2 & unit8 & unit5 & \cellcolor{red!20}{unitE} & unit7 & unit6 & unit3 &  \\
\hline
\end{tabular}
\vspace{0.5mm}
\caption{TDMA schedule for the instantiated FlexRay network.}
\label{tab:nwschedule}
\end{table}

The network parameters in OMNeT++ are configured using their default values and units\footnote[1]{The units are abbreviated as follows: $\mu s$ – microseconds, $ns$ – nanoseconds, $ps$ – picoseconds, $MT$ – macroticks, $\mu T$ – microticks.}, except in cases where Table~\ref{tab:nwparams} provides a different choice. For completeness, the table also includes default parameters that, while unchanged, are relevant to the discussion in this section.

\begin{table}[h]
\centering
\arrayrulecolor{gray} 
\setlength{\arrayrulewidth}{0.1mm}
\setlength{\tabcolsep}{2.25pt}
\renewcommand{\arraystretch}{0.9}
\begin{tabular}{|l|l|l|}
\hline
\textbf{Network parameter} & \textbf{Value} & \textbf{Reasoning} \\
\hline
pdMicrotick & 12.5 $ns$ & FiCo4OMNeT default \\
\hline
gdMacrotick & 1 $\mu s$ & FiCo4OMNeT default \\
\hline
gdStaticSlot & 4 $MT$& FiCo4OMNeT default \\
\hline
gdNIT & 2 $MT$& FiCo4OMNeT default \\
\hline
numberOfNodesChannelA & 8 & Number of nodes, \\
 & & $N\geq3F+1$, \\
\cline{1-2}
numberOfNodesChannelB & 8 & $F=1$ faulty node\\ 
 & & $\rightarrow{} N \geq 4$  \\
\hline
gNumberOfStaticSlots & 8 & to provide each node\\
 & & in the communication cycle\\
 & & with a designated slot \\
\hline
gCycleCountMax & 7 & smallest value configurable \\
 & & in FiCo4OMNeT for a \\
 & & concise TDMA schedule\\
\hline
pOffsetCorrectionOut\textsuperscript{\ref{fn:recalc_param}} & 80 $\mu T$ & the maximum offset \\
 & & correction is limited \\
 & & by the maximum clock \\
 & & deviation between the \\
 & & internally synchronized \\
 & & node clocks, \PiNN{}\\
\hline
maxDrift & 18.75 $ps$ & maximum clock deviation \\ 
 & & due to drift between\\
 & & two non-synchronized clocks\\
 & &  in a single microtick.\\
\hline
maxDriftChange & 1 $ps$ & based on example \\
 & & network \textit{syncnodes} from \\
 & & FiCo4OMNeT\\
\hline
gdActionPointOffset & 2 $MT$ & to ensure transmission \\ 
 & & window alignment \cite{Bauer03_TWA}, for\\
 & & $\PiNN{} = \text{1 }MT$.\\
\hline
channel.delay\textsuperscript{\ref{fn:recalc_param}} & 127.94 $ns$ & sets maximum propagation\\
 & & delay for deriving the \\
 & & parameter aAssumedPrecision\\
\hline
\end{tabular}
\vspace{1.0mm}
\caption{Network configuration parameters for the simulated network model.}
\label{tab:nwparams}
\end{table}

The parameter \textit{maxDrift} is defined by the FiCo4OMNeT simulation model to simulate the clock drift phenomenon. We arrive at the configured value using \textit{pdMicrotick} and the maximum clock frequency deviation considered for FlexRay, \textit{cClockDeviationMax} (0.0015 $s/s$) \cite{ProtocolSpecfication_Flexray_3_0_1}. The parameter, \textit{pdmaxDrift} (not to be confused with \textit{maxDrift}), refers to the maximum clock deviation between two nodes with non-synchronized clocks for one communication cycle, due to clock drift alone. It is also referred to as \textit{Drift Offset} ($\Gamma$) \cite{DBLP:KopetzS22_Clocks_Time}. 

Table~\ref{tab:derparams} lists key network parameters derived from the configured network parameters, connecting the network simulation to the theoretical understanding of a time-triggered network.
To set \PiNN{} in FlexRay, the parameter \emph{aAssumedPrecision} must be configured. This parameter, given in Table~\ref{tab:derparams}, can be derived using the largest microtick and the total propagation delay, \emph{adPropagationDelayMax} (See equations 11 and 13 in the FlexRay's protocol specifications\cite{ProtocolSpecfication_Flexray_3_0_1}).
\emph{adPropagationDelayMax} is indirectly configured using the parameter \textit{channel.delay}\footnote[2]{This value is recalculated for every configuration in experiment~\ref{subsec:test_tinj} \label{fn:recalc_param}}. We assume $aAssumedPrecision = \text{1 } MT$ \textsuperscript{\ref{fn:recalc_param}} as the default and maximum value, to meet the reasonableness condition of the global time base \cite{DBLP:KopetzS22_Clocks_Time}. We assume $\PiNG{} = \PiNN{}$, which represents the best-case scenario for network-agnostic guardians, as shown in the proof in Section~\ref{sec:ngprecision}.

\begin{table}[h]
\centering
\arrayrulecolor{gray}
\setlength{\arrayrulewidth}{0.2mm}
\setlength{\tabcolsep}{2.25pt}
\renewcommand{\arraystretch}{1.0}
\begin{tabular}{|l|l|l|}
\hline
\textbf{Network parameter} & \textbf{Basis} & \textbf{Value} \\
\hline
gClockDeviationMax(\maxDrift{}) & cClockDeviationMax & 0.0015 $s/s$ \\
\hline 
aMicroPerMacroNom & \multirow{2}*{\normalsize $\frac{\text{gdMacrotick}}{\text{pdMicrotick}}$}  & 80 $\mu T$ \\
($MT/uT$) & & \\
\hline 
pMicroPerCycle(\Resync{}) & (gdNIT + gdStaticSlot $\times$ & 2720 $\mu T$ \\
 & gNumberOfStaticSlots) & \\
 & $\times$ aMicroPerMacroNom & \\
\hline 
pdmaxDrift/ Drift Offset & ceil(2$\times \frac{\text{gClockDeviationMax}}{\text{(1-gClockDeviationMax))}}$ & 9 $\mu T$ \\
($\Gamma = \lceil 2\maxDrift{}\Resync{} \rceil$) & $\times$ gMacroperCycle)& \\
\hline
adPropagationDelayMax\textsuperscript{\ref{fn:recalc_param}} & 2 $\times$ channel.delay & 255.88 $ns$ \\
\hline 
aAssumedPrecision & See Section~\ref{sec:evaluation} & 1 $MT$ \\
(\PiNN{}) & & \\
\hline
\end{tabular}
\vspace{1.0mm}
\caption{Derived network parameters (using Table~\ref{tab:nwparams}) for the simulated network model.}
\label{tab:derparams}
\end{table}

Each experiment comprises multiple simulation configurations, with each simulation running for 147 communication cycles (sim-time-limit~\cite{omnetppsimmanual}) and repeated 500 times. The configuration parameters are specified using their default units or as a fraction of \PiNN{}.

In Section~\ref{subsec:test_tinj}, we examine how the clock synchronization precision between the attacker node and its guardian, \PiNG{}, influences the covertness of the target slot attack. In Section~\ref{subsec:test_ntc}, we evaluate the attack procedure from Section~\ref{subsec:attack_proc} for every possible victim node for our network model. The relative slot position of the victim node with respect to the attacker's slot determines the clock deviation required for slot overlap. Taking this into account, and assuming a fixed \PiNG{} along with a random clock drift drawn from the interval $[0,\maxDrift{}]$ at the beginning of each communication cycle, we evaluate whether the attacker node can successfully achieve slot overlap with the victim node.

\subsection{Maximum Covertly Injected Delay}\label{subsec:test_tinj}

The purpose of the first experiment is two-fold:
\begin{enumerate}
    \item to validate, through simulation, the theoretical expression in Equation~\ref{eq:tnj_max}, which represents the maximum delay that can be injected covertly with the CS signal from the node to its guardian,
    \item to show that, attacks using injected delay greater than maximum covert injected delay can remain undetected with varying success, depending on the choice of the precision for node-guardian clock synchronization.
\end{enumerate}

We evaluate three simulation configurations by setting three values of \PiNN{}, and correspondingly \PiNG{}, and then varying \tinj{} in the range $[0,\PiNG{}+\maxDrift{}\Resync{}]$ for each value of \PiNG{}. The upper limit of the range of injected delay accounts for either the node or guardian local clock drifting maximally, while reflecting the low probability of both clocks drifting maximally in opposite directions. 

According to the constraints for aAssumedPrecision \cite{ProtocolSpecfication_Flexray_3_0_1}, if the timing inaccuracies of the physical layer are considered negligible ($adPropagationDelayMax = 0$), $aAssumedPrecision \approx \text{0.5 }MT$. Consequently, we assume $aAssumedPrecision \geq 0.5MT$ for the configurations, though the specific values selected are representative.

The precision of internal clock synchronization, \PiNN{} depends on~\cite{DBLP:KopetzS22_Clocks_Time},
\begin{itemize}
    \item the maximum divergence of any two clocks, or drift offset $\Gamma$,
    \item the convergence function, $\phi$ which, in turn, depends on Byzantine error and the latency jitter of synchronization messages due to physical layer timing inaccuracies, such as propagation delay.
\end{itemize}
The maximum covert injected delay, as discussed in Section~\ref{subsec:calc_tinj}, is calculated using the precision of node-guardian clock synchronization, \PiNG{} and the drift offset, $\Gamma$. However, since $\Gamma$ influences the choice of \PiNG{}, we set the test configurations directly for different values of \PiNG{}.

We set the values of \PiNN{}, and correspondingly \PiNG{}, by varying the maximum propagation delay using the \textit{channel.delay} parameter. The required configuration values are given in Table~\ref{tab:test1_config}. We adapt the value of \textit{pOffsetCorrectionOut} for each configuration accordingly.

\begin{table}[h!]
\centering
\setlength{\arrayrulewidth}{0.5mm}
\setlength{\tabcolsep}{3pt}
\renewcommand{\arraystretch}{0.8}
\begin{tabular}{|p{1.3cm}|p{3.2cm}|p{2.5cm}|}
\hline
\textbf{Config. no.} & \textbf{aAssumedPrecision}/$\bm{\PiNG{}}$ & \textbf{channel.delay (ns)} \\ \hline
1 & 1 $MT$/ 80 $\mu T$ & 127.94 \\ \hline
2 & 0.75 $MT$/ 60 $\mu T$ & 65.44 \\ \hline
3 & 0.5 $MT$/ 40 $\mu T$ & 2.94 \\ \hline
\end{tabular}
\caption{Network parameters for Experiment~\ref{subsec:test_tinj} configurations.}
\label{tab:test1_config}
\end{table}

For each configuration, we report the \emph{percentage of undetected attacks}, which we define as the ratio of simulation repetitions in which the attack remains undetected to the total number of repetitions (500). 

Figure~\ref{fig:test_1a} depicts the results for the first configuration, with precision $\PiNG{} = \textit{1 macrotick (MT)}$. Similarly, the figures~\ref{fig:test_1b}~and~\ref{fig:test_1c}, depict the results for $\PiNG{} = \textit{0.75 MT}$ and $\PiNG{} = \textit{0.5 MT}$, respectively. The maximum covert injected delay calculated for each value of \PiNG{} has been indicated for reference. The results indicate that as long as \tinj{} complies with Eq.~\ref{eq:tnj_max}, \emph{unitE} remains undetected. Further, an attacker may inject a delay larger than the maximum covert value and still remain undetected. For clock deviation between the node and the guardian smaller than $\Gamma$, \tinj{} values exceeding the theoretical maximum can still result in a CS signal that passes the acceptance window check, since Eq.~\ref{eq:tnj_max}, being pessimistic, assumes the maximum possible clock drift between the node and the guardian. The percentage of undetected attacks, for a particular value of injected delay greater than the maximum, depends on \PiNG{}.

The plots in Figures~\ref{fig:test_1a} through \ref{fig:test_1c}, evidently, do not depict strictly decreasing functions. Though this behavior is only observed for the plots in Figures~\ref{fig:test_1a}~and~\ref{fig:test_1c}, for each value of \PiNG{}, there is a small interval of \tinj{} in the approximate range of $[0.97,1.02] \times \PiNG{}$ for our experiment where the percentage of undetected attacks may increase. The reason lies in the clock deviation calculation at the guardian during clock synchronization. 
To simulate a real-time clock, 
we convert the OMNeT++ simulation time to guardian local clock microticks using a ceiling operation, introducing a possible error of $\pm 1$ microtick. Considering the randomness of simulated clock drift and the ceiling operation error, the node and guardian clock deviation can reach or exceed \PiNG{}, for \tinj{} in the range $[0.97,1.02]\times \PiNG{}$. This can result in a short spike in the percentage of undetected attacks during the interval.

\pgfplotsset{width=7.5cm,compat=1.16}

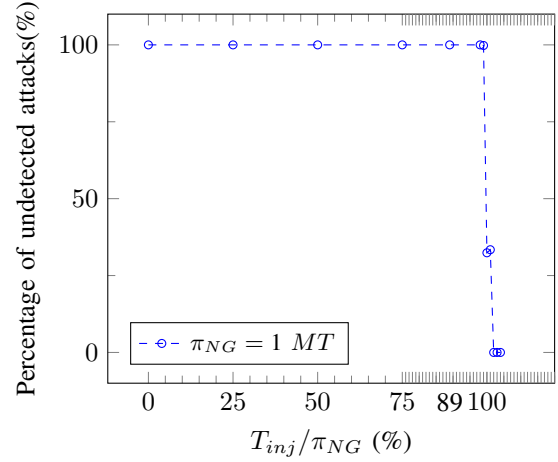
\begin{figure}[h]
\centering
\begin{tikzpicture}
  \begin{axis}[
    xmax=120,
    xlabel={$T_{inj}/\pi_{NG}$ (\%)},
    ylabel={Percentage of undetected attacks(\%)},
    xtick={0,25,50,75,100,125},
    minor x tick num=4,
    minor y tick num=1,
    extra x ticks={76,77,78,79,80,81,82,83,84,85,86,87,88,89,90,91,92,93,94,95,96,97,98,99,100,101,102,103,104,105,106,107,108,109,110,111,112,113,114,115,116,117,118,119,120},
    extra x tick labels={,,,,,,,,,,,,,89,,,,,,,,,,,,,,,,,,,,,,,,,,,,,,,},
    legend style={
        at={(0.05,0.05)},
        anchor=south west,
        font=\small,           
        scale=0.95              
        },
  ]

    \addplot[
      blue,
      mark=o,
      dashed,
      mark size=1.5pt,
      mark options={solid}
    ] table[
      col sep=comma,
      x=T_injected,
      y={Percentage of undetected attacks}
    ]{Stats_PI_v_Tinj_Pi=1MT.csv};
    \addlegendentry{$\PiNG{}=\text{1 }MT$}

  \end{axis}
\end{tikzpicture}
\caption{For $\PiNG{}=\text{1 }MT$, percentage of undetected attacks for increasing values of injected delay.}
\label{fig:test_1a}
\end{figure}

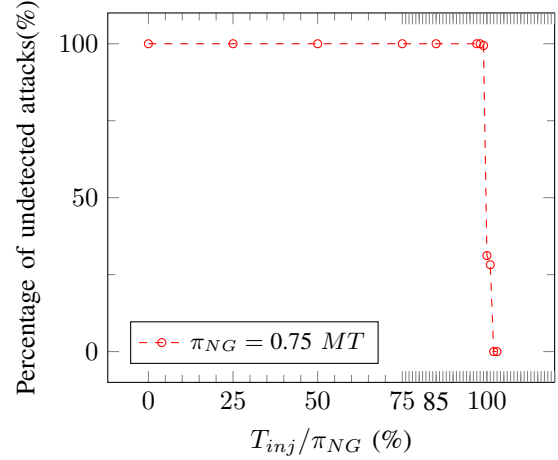
\begin{figure}[h!]
\centering
\begin{tikzpicture}
  \begin{axis}[
    xmax=120,
    xlabel={$T_{inj}/\pi_{NG}$ (\%)},
    ylabel={Percentage of undetected attacks(\%)},
    xtick={0,25,50,75,100,125},
    minor x tick num=4,
    minor y tick num=1,
    extra x ticks={76,77,78,79,80,81,82,83,84,85,86,87,88,89,90,91,92,93,94,95,96,97,98,99,100,101,102,103,104,105,106,107,108,109,110,111,112,113,114,115,116,117,118,119,120},
    extra x tick labels={,,,,,,,,,85,,,,,,,,,,,,,,,,,,,,,,,,,,,,,,,,,,,},
    legend style={
        at={(0.05,0.05)},
        anchor=south west,
        font=\small,           
        scale=0.95              
        },
  ]

    \addplot[
      red,
      mark=o,
      dashed,
      mark size=1.5pt,
      mark options={solid}
    ] table[
      col sep=comma,
      x=T_injected,
      y={Percentage of undetected attacks}
    ]{Stats_PI_v_Tinj_Pi=0.75MT.csv};
    \addlegendentry{$\PiNG{}=\text{0.75 }MT$}

  \end{axis}
\end{tikzpicture}
\caption{For $\PiNG{}=\text{0.75 }MT$, percentage of undetected attacks for increasing values of injected delay.}
\label{fig:test_1b}
\vspace{-0.6cm}\end{figure}

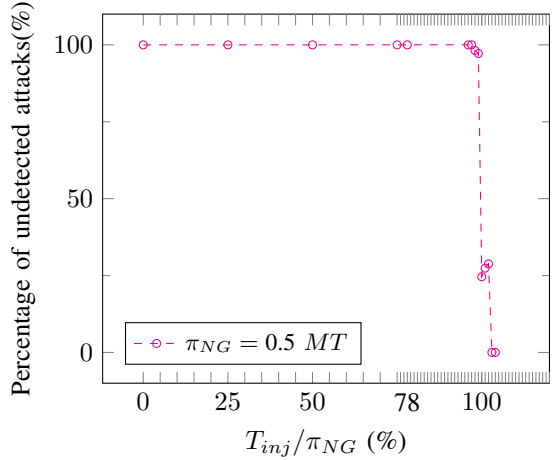
\begin{figure}[h!]
\centering
\begin{tikzpicture}
  \begin{axis}[
    xmax=120,
    xlabel={$T_{inj}/\pi_{NG}$ (\%)},
    ylabel={Percentage of undetected attacks(\%)},
    xtick={0,25,50,100,125},
    minor x tick num=4,
    minor y tick num=1,
    extra x ticks={5,10,15,20,30,35,40,45,55,60,65,70,75,76,77,78,79,80,81,82,83,84,85,86,87,88,89,90,91,92,93,94,95,96,97,98,99,100,101,102,103,104,105,106,107,108,109,110,111,112,113,114,115,116,117,118,119,120},
    extra x tick labels={,,,,,,,,,,,,,,,78,,,,,,,,,,,,,,,,,,,,,,,,,,,,,,,,,,,,,,,,,,},
    legend style={
        at={(0.05,0.05)},
        anchor=south west,
        font=\small,           
        scale=0.95              
        },
  ]

    \addplot[
      magenta,
      mark=o,
      dashed,
      mark size=1.5pt,
      mark options={solid}
    ] table[
      col sep=comma,
      x=T_injected,
      y={Percentage of undetected attacks}
    ]{Stats_PI_v_Tinj_Pi=0.5MT.csv};
    \addlegendentry{$\PiNG{}=\text{0.5 }MT$}

  \end{axis}
\end{tikzpicture}
\caption{For $\PiNG{}=\text{0.5 }MT$, percentage of undetected attacks for increasing values of injected delay.}
\label{fig:test_1c}
\end{figure}

\subsection{Required Clock Deviation for Slot Overlap}\label{subsec:test_ntc}

The purpose of the second experiment is to demonstrate successful target slot attacks that remain undetected and invalidate specific target frames through collision with messages from the attacker.
We test seven simulation configurations, and in each configuration \emph{unitE} targets a specific slot in the schedule. 

During each simulation, with a specific victim node, \emph{unitE} (Slot $5$) statically calculates the required clock deviation for slot overlap with the target slot and performs acceptance window attacks against its guardian using the maximum covert injected delay. Considering the values for \PiNN{} (aAssumedPrecision) and Drift Offset ($\Gamma$), the maximum covert injected delay is $0.89\times\PiNG{}$. In each communication cycle, \emph{unitE} calculates the clock deviation from the victim node using the received message and compares it to the required clock deviation for slot overlap. If overlap is identified, \emph{unitE} transmits its message during the scheduled slot to invalidate the critical message sent during the target slot. 

We consider the attack successful when it disrupts the critical message transmitted during the target slot. Table~\ref{tab:test2} summarizes the results, showing the \emph{percentage of successful attacks}, calculated as the ratio of repetitions with successful disruption of the critical message to the total number of repetitions (500) for each selected victim node.

\begin{table}[h!]
\centering
\setlength{\arrayrulewidth}{0.2mm}
\setlength{\tabcolsep}{3.2pt}
\renewcommand{\arraystretch}{0.8}
\begin{tabular}{|p{1cm}|p{1.2cm}|p{2cm}|p{2.3cm}|}
\hline
\textbf{Config. no.} & \textbf{Target slot} & \textbf{Required clock deviation (MT)} &  \textbf{\textit{successful} attacks (\%)}   \\
\hline
1 & 1 & 15 & 100\% \\
\hline
2 & 2 & 11 & 100\% \\
\hline
3 & 3 & 7 & 100\% \\
\hline
4 & 4 & 3  & 100\% \\
\hline
5 & 6 & 5 & 100\% \\
\hline
6 & 7 & 9 & 100\% \\
\hline
7 & 8 & 13 & 100\% \\
\hline
\end{tabular}
\caption{Results for the required clock deviation and percentage of successful target slot attacks.}
\label{tab:test2}
\vspace{-0.5cm}
\end{table}
The results indicate that, for each target slot (7 configurations) during every repetition (500), \emph{unitE} identifies slot overlap accurately despite random clock drift and its message causes a collision with the target frame. 

The second experiment supports the applicability of the attack parameters maximum injected delay, the required clock deviation for slot overlap (at the time of initiation of the attack), and the attack procedure successfully employing both in mounting target slot attacks. When the parameters are appropriately chosen, the attack remains undetected by the guardian, continuously injects covert delay and monitors the relative clock deviation with respect to the target until a slot overlap is identified, which makes the collision with the target critical frame certain, as explained in Section~\ref{sec:attack}.

\section{Countermeasures}\label{sec:countermeasures}
The chosen clock‑synchronization precision between a node and its network‑agnostic guardian does not prevent targeted‑slot attacks, as demonstrated by the theoretical result in Section~\ref{sec:ngprecision}. We complement this analytic finding with the simulation study in Section~\ref{subsec:test_tinj}, which evaluates the attack’s covertness as a function of the injected-delay parameter \tinj{} for a fixed \PiNN{}.

Alkoudsi et al.~\cite{syncguard2023, ttarejuv2022} investigated the coordination of proactive and reactive node rejuvenation as a mechanism to maintain synchrony among network-agnostic guardians in the presence of malicious nodes. 
Assuming a secure start-up service, guardian-enforced rejuvenation establishes an initial \emph{interval of correctness} post node rejuvenation. 
During this interval, nodes are deemed trustworthy, enabling guardians to safely synchronize their clock offsets. Consequently, periodic proactive node rejuvenation systematically eliminates the presence and effects of an attacker. Simultaneously, the exchange of authenticated messages and the utilization of a membership protocol allow guardians to detect and reactively rejuvenate compromised nodes, albeit at the cost of increased communication bandwidth.

In the specific context of the targeted slot attack, which relies on a cumulative injection effect to gradually desynchronize guardians, node rejuvenation serves a critical role:   
Proactive rejuvenation \cite{DBLP:conf/sac/SousaNV06, ttarejuv2022} effectively constrains the attacker's ability to align with a victim slot by upper-bounding the time available for an attacker to accumulate drift. 
Furthermore, having guardians exchange and cross-verify system-level data, such as TDMA positioning and membership information enables guardians to identify compromised nodes and enforce their reactive rejuvenation.

We suspect that proactive rejuvenation alone may suffice for certain system configurations, depending on the structure of the communication schedule, the clock synchronization protocol, and the time required to compromise a node. 
Relying solely on proactive measures could facilitate the design of more efficient, yet still time-domain attack-resilient, network-agnostic guardians. 

\section{Conclusion}\label{sec:conclusion}
Network-agnostic guardians are trusted components that enforce predetermined schedules and ensure temporal fault containment, while offering significant advantages over network-aware variants in terms of cost, complexity and scalability. However, a major caveat is their reliance on the controlled node for clock synchronization.

In this paper, we advanced the understanding of timing attacks against network-agnostic guardians by presenting a target slot attack that invalidates a specific message in the time-triggered schedule.

We established a lower bound on the precision of external clock synchronization between a node and its network-agnostic guardian, as well as identified the maximum delay that can be covertly injected alongside the node's clock synchronization signal to its guardian. 
By systematically injecting such delays, the adversary can alter the guardian's perception of global time until its designated slot aligns with the target slot, 
resulting in the malicious node's message colliding with and invalidating the critical message sent during the target slot. 
Simulation results from an eight-node FlexRay network validate the attack's effectiveness with 100\% success and demonstrate its significant impact on message delivery within time-triggered communication protocols that use network-agnostic guardians.

The identified vulnerability of network-agnostic guardians to timing attacks underscores the importance of coordination among guardians. 
Since communication between guardians is costly, 
we will explore, in future work, the interplay between enforcing proactive node rejuvenation by network-agnostic guardians and the necessity for explicit communication among guardians.

\bibliographystyle{plain} 
\bibliography{internal}	

\end{document}